\documentclass[journal,12pt,onecolumn,draftclsnofoot]{IEEEtran}
\usepackage[utf8]{inputenc} 
\usepackage[T1]{fontenc}

\usepackage{ifpdf}
\usepackage{cite}
\ifCLASSINFOpdf
  \usepackage[pdftex]{graphicx}
\else
\fi
\usepackage[cmex10]{amsmath}
\usepackage{amsfonts}
\usepackage{amssymb}
\usepackage{amsthm}
\usepackage{color}
\usepackage{algorithmic}
\usepackage{array}
\usepackage{url}
\newtheorem{theorem}{Theorem}

\newtheorem{corollary}{Corollary}

\usepackage{amssymb}
\usepackage{amscd}
\usepackage{delarray}
\usepackage{color}

\def\Reals{\mathop{\hbox{\mit I\kern-.2em R}}\nolimits}

\def\RA{\mathop{\hbox{\Rightarrow\kern-.2em /}}}

\def\Complexes{\mathop{\hbox{\mit C\kern-.44em
               \vrule depth 0ex height 1.4ex width .06em
               \kern.41em}}\nolimits}

\def\Zeals{\mathop{\hbox{\mit Z\kern-.29em Z}}\nolimits}

\def\Neals{\mathop{\hbox{\mit I\kern-.2em N}}\nolimits}

\def\K{\mathop{\hbox{\mit I\kern-.2em K}}\nolimits}

\def\F{\mathop{\hbox{\mit I\kern-.2em F}}\nolimits}

\def\sqr#1#2{{\vcenter{\vbox{\hrule height.#2pt
  \hbox{\vrule width.#2pt height#1pt \kern#1pt
    \vrule width.#2pt}
   \hrule height.2pt}}}}

\def\rddots{\mathinner{\mkern1mu\raise1pt\vbox{\kern7pt\hbox{.}}\mkern2mu\raise4pt\hbox{.}\mkern2mu\raise7pt\hbox{.}\mkern1mu}}

\usepackage[dvips]{epsfig}

\begin{document}

\title{\sc Lattice Erasure  Codes of Low Rank with Noise Margins}
\author{\IEEEauthorblockN{Vinay A. Vaishampayan}\\
\IEEEauthorblockA{Dept. of Engineering Science and Physics\\City University of New York-College of Staten Island\\Staten Island, NY USA}
}

\maketitle
\begin{abstract}
 We consider the following generalization of an $(n,k)$ MDS code  for application to an erasure channel with additive noise. Like an MDS code, our code is required to be decodable from any $k$ received symbols, in the absence of noise. In addition, we require that the noise margin for every allowable erasure pattern be as large as possible and that the code satisfy a power constraint. In this paper we derive performance bounds and present a few designs for low rank lattice codes for an additive noise channel with erasures.  

 \end{abstract}
 
\section{Introduction}
 \begin{figure}[h] 
   \centering
   \includegraphics[width=7cm]{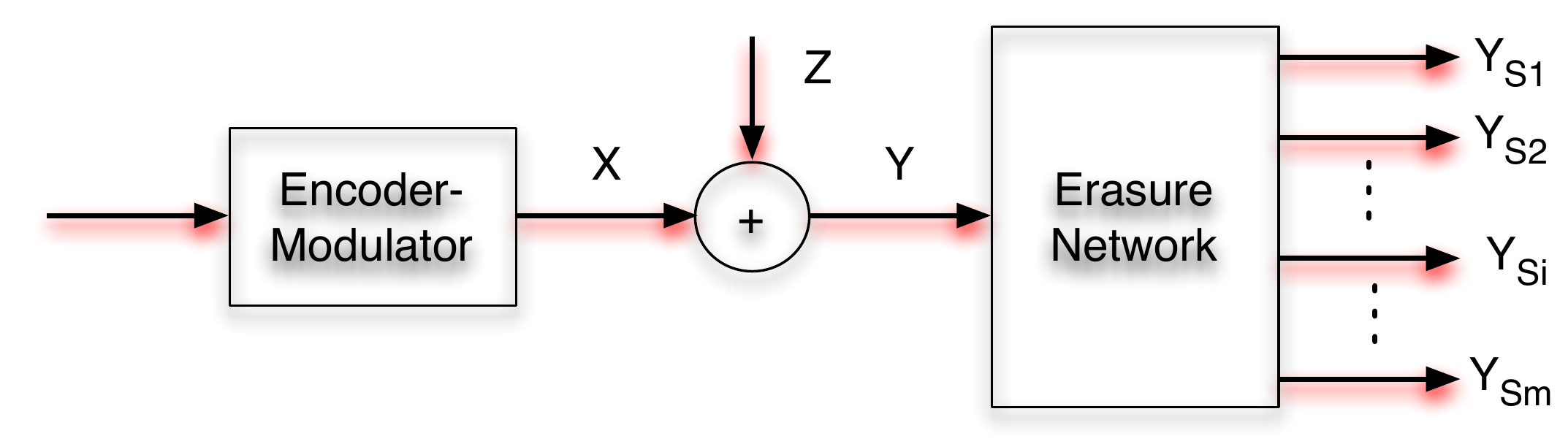} 
   \caption{Coding and Modulation for the Erasure Network.}
   \label{fig:enetwork}
\end{figure}

We consider  low rank lattice codes for transmission over a noisy erasure channel as  illustrated in Fig.~\ref{fig:enetwork}. In this figure $k$ information symbols are mapped by an  encoder/modulator to a vector $x=(x_1,x_2,\ldots,x_n) \in  \Lambda$, where $\Lambda$ is a rank-$k$ lattice in $\mathbb{R}^n$. The output of the additive noise channel is $y=x+z$, where $z=(z_1,z_2,\ldots,z_n)$ is a noise vector independent  of $x$ and with independent components.   Components of $y$ are then erased by an erasure network, whose outputs  are obtained by retaining only those symbols of $y$  indexed by subsets $S\subset \{1,2,\ldots,n\}$ in a given sub-collection of subsets; thus $y_S$ coincides with $y$ is the positions identified by $S$. As an example, with $n=4$, $S=\{2,4\}$ and $y=(a,b,c,d)$, $y_S=(b,d)$. A decoder estimates the source symbols based on  $y_S$ with a probability of error denoted $P_e(S)$. The objective is to minimize $P_e(S)$ for each $S$ by designing a single codebook which satisfies a power constraint $E [X^tX] \leq nP$, where $E$ denotes expectation with respect to a uniform distribution on the codebook. Here we consider as our sub-collection $\mathcal S$, \emph{all $k$-subsets} of $\{1,2,\ldots,n\}$. Our paper is organized as follows. Prior work and lattice background is in Sec.~\ref{sec:previous}. Two performance bounds are presented in Sec.~\ref{sec:bounds}, constructions for codes in dimension $n=4$ are presented and compared to the derived bounds  in Sec.~\ref{sec:fourdimcodes}. A summary is  in Sec.~\ref{sec:summary}.

We use the acronym w.l.o.g to mean `without loss of generality'.

\section{Prior Work and Review of Lattice Terminology}
\label{sec:previous}
The problem considered here may be viewed as a  code design problem for a special case of the \emph{compound channel}, see e.g.  ~\cite{blackwell1959capacity}, \cite{csiszar1991capacity}. This work was motivated by a study on cross layer coding that appeared in~\cite{courtade2011optimal}. For prior contributions on the Gaussian erasure channel, please refer to \cite{ozccelikkale2014unitary} and the references therein.We now develop  notation and some basic definitions for low rank lattices in $\mathbb{R}^n$.
Let $\{\phi_i,~i=1,2,\ldots,k\}$ be a collection of $k \leq n$ orthonormal column vectors in $\mathbb{R}^n$ and let $\Phi=(\phi_i,~i=1,2,\ldots,k)$ denote the associated  $n\times k$ orthonormal matrix. Let $V$ denote a $k \times k$ generator matrix of full rank for a lattice $\Lambda_V=V\mathbb{Z}^k:=\{Vu,~u\in \mathbb{Z}^k\}$.  We will refer to $\Lambda_V$ as the \emph{mother} lattice. Let 
\begin{equation}
\Lambda=\Phi \Lambda_V:=\Phi V \mathbb{Z}^k.
\label{eqn:defLambda}
\end{equation}
$\Lambda$ is a rank-$k$ lattice in $\mathbb{R}^n$. 
Let $G(\Lambda_V)=V^tV$ denote the Gram matrix of $\Lambda_V$ ($^t$ is the transpose operator). 

The \emph{determinant} of a lattice $\det \Lambda$ is defined in terms of the determinant of its Gram matrix by
$
\det \Lambda:=det(G(\Lambda)).
$
Let   $\rho(\Lambda)$   denote the radius of the largest inscribed sphere in a Voronoi cell of $\Lambda$. The packing density of $\Lambda$ is defined in terms of $V_k$ the volume of a unit-radius Euclidean ball in $\mathbb{R}^k$ by
\begin{equation}
\Delta_k(\Lambda)=V_k \rho^k/\sqrt{\det\Lambda}.
\end{equation}
We denote by $\Delta_k(opt)$, the largest packing density that can be acheived by any lattice in $\mathbb{R}^k$. The problem of finding lattices that maximize the packing density is a classical problem in number theory and geometry, with several excellent references~\cite{SPLAG},~\cite{gruber1987geometry}.

The following definitions are from~\cite{gruber1987geometry}.  A body captures the notion of a solid subset of $\mathbb{R}^n$, specifically,   $B\subset \mathbb{R}^n$ is a \emph{body} if it has nonempty interior and is contained in the closure of its interior~\cite{gruber1987geometry}. 
A body $B\subset \mathbb{R}^n$ is said to be \emph{centrally symmetric} if $B=-B$, where $-B=\{-x~:~x \in B\}$. A closed body $B$ with the property that for any $x \in B$, the point $\lambda x \in B$ for every $0 \leq \lambda < 1$ is called a \emph{star body}.  While convex bodies are star bodies, the converse is not true. A simple example, and one directly relevant to us is the star body formed by the union of centrally symmetric ellipsoids in $\mathbb{R}^n$. A lattice $\Lambda$ is said to be \emph{admissible} for $B \subset \mathbb{R}^n$, or $B$-admissible, if no non-zero point in $\Lambda$ lies in  $B$. The greatest lower bound of $\sqrt(\det \Lambda)$ over all $B-$admissible lattices is called the lattice constant of $B$, denoted $\Delta(B)$ (which is set to $\infty$ is there are no $B-$ admissible lattices).  A  $B$-admissible lattice $\Lambda$ with $\det \Lambda=\Delta(B)^2$ is said to be a \emph{critical} lattice for $B$.  

A lattice $\Lambda$ is said to be a \emph{packing lattice} for a body $B$ if the sets $B$ and $B+\lambda$ are disjoint for all non-zero $\lambda \in \Lambda$. It is known, Thm.1, Ch. 3, Sec 20 \cite{gruber1987geometry}, that $\Lambda$ is a packing lattice for centrally symmetric, convex body $B$ if and only if it is admissible for $2B$. Thus, for a convex body, the problem of finding a packing lattice for $B$ is equivalent to that of finding an admissible lattice for $2B$. The connection between packings and admissibility for non-convex bodies is messier.  The distinction arises because for a centrally symmetric body $B$, $\Lambda$ is  a lattice packing of $B$ if and only if it is admissible for $B+B$, where $+$ denotes the set sum or Minkowski sum. If the centrally symmetric body is also convex, then $B+B=2B$ and thus packing problems and admissibility problems are closely related. On the other hand, if $B$ is centrally symmetric but non-convex, in order to solve a packing problem for $B$ one must solve an admissibility problem for $B+B$, and this set may not be as easily described as $B$.

 In our application, we need to index body $B$ by subset $S$ in a given sub-collection of subsets and our problem is one of packing $\bigcup_{S\in {\mathcal S}}B(S)$, which is non-convex. While admissibility for non convex centrally symmetric body  $2C$ says nothing in general about packings for $C$, it turns out that  our design problem  is equivalent to finding a critical lattice for $2\bigcup_{S\in {\mathcal S}}B(S)$ because the decoder  knows $S$.  Thus it is possible to draw on the theory of admissible lattices for star bodies. This theory provides several key ingredients to help find a solution to this problem.  Most notably,  in the chapter on Mahler's compactness theorem~\cite{cassels2012introduction}, Theorem VII states that \emph{every critical lattice for a bounded star body $\mathcal S$ has $n$ linearly independent points on the boundary of $S$.}  
 
\section{Bounds}
\label{sec:bounds}
Let $\mathcal{I}_n=\{1,2,\ldots,n\}$. For $S \subset \mathcal{I}_n$, $|S|=k$ let $\Lambda_S$ be the lattice obtained be retaining only those coordinates that are in $S$ or equivalently $\Lambda_S$ is the projection of $\Lambda$ into the subspace ${\mathcal C}_S :=\mbox{Span}\{e_i,~i\in S\}$ where  $e_i=(0,...,0,1,0,...,0)^t$ is the $i$th unit vector in $\mathbb{R}^n$. For any k-subset $S\subset \{1,2,\ldots,n\}$, we denote by $\Phi_S$ the $k \times k $ submatrix obtained by extracting from $\Phi$ the $k$ rows identified by $S$. The generator matrix for $\Lambda_S$ is $\Phi_S V$ and its Gram matrix 
$G(\Lambda_S)=V^t\Phi_S^t \Phi_SV$. 

Define the (packing volume) contraction ratio
\begin{equation}
\beta_S=(\rho(\Lambda_S)/\rho(\Lambda_V))^k
\end{equation}
let $\rho_{min}=\min_S \rho(\Lambda_S)$ and let $\beta_{min}=\min_S \beta_S$.

\subsection{Determinant Upper Bound}
We will use symbols $\bar{x}$, $x^\#$ to denote the arithmetic mean and geometric mean, respectively, of the real numbers
$x_i$ over some index set $\mathcal I$.
When a $k$-dim mother lattice $\Lambda_V$ is set in $\Reals^n$ using a basis $\Phi$, the projections on the ${n \choose k}$ subsets $S$, cannot all be simultaneously good. There are two important factors that measure the `goodness' of the projections---the packing density and the scale of the lattices $\Lambda_S$. The following theorem develops one of two bounds presented in this paper.
\begin{theorem} ({\bf Determinant Bound})
Given a mother lattice $\Lambda_V$ and orthonormal basis $\Phi$, let $\beta^\#$ and $\Delta^\#$ be respectively, the geometric mean  of the volume contraction ratios and packing densities of the child lattices $\Lambda_S$, taken  over all $k$-subsets of $\{1,2,\ldots,n\}$. Then
\begin{equation}
(\beta^\# \Delta(\Lambda_V))^2 \leq \frac{(\Delta^\#)^2}{{n \choose k}}.
\label{eqn:DetBoundTh}
\end{equation}
Equality holds if and only if all child lattices have equal determinants.
\end{theorem}
\begin{proof}
The packing densities of the mother lattice $\Lambda_V$ and child lattice $\Lambda_S$ are related by the following identity
\begin{equation}
\Delta^2(\Lambda_V) \beta_S^2=\Delta^2(\Lambda_S) \frac{\det \Lambda_S}{\det \Lambda_V}.
\end{equation}
Compute the geometric mean of both sides over the collection of $k$-subsets $S$ to get
\begin{equation}
\Delta^2(\Lambda_V) (\beta^\#)^2=(\Delta^\#)^2 \left(\prod_S\frac{\det \Lambda_S}{\det \Lambda_V}\right)^\frac{1}{{n \choose k}}.
\end{equation}
From the arithmetic-geometric mean inequality it follows that
\begin{equation}
\Delta^2(\Lambda_V) (\beta^\#)^2 \leq (\Delta^\#)^2 \frac{1}{{n \choose k}}\left(\sum_S\frac{\det \Lambda_S}{\det \Lambda_V}\right)
\end{equation}
and equality holds if and only if $\det \Lambda_S$ is a constant with respect to $S$.
However 
\begin{eqnarray}
\sum_S \det \Lambda_S & = & \sum_S \det(G(\Lambda_S)) \nonumber \\
& =  & \sum_S \det ((\Phi_S V)^t) \det(\Phi_S V) \nonumber \\
& \stackrel{(a)}{=} & \det( (\Phi V)^t (\Phi V)) \nonumber \\
& = & \det \Lambda_V,
\end{eqnarray}
where in (a) we have used the Cauchy-Binet formula, see e.g.~\cite{horn2012matrix}. The remainder of the proof follows directly.
\end{proof}
\noindent
The following corollary is immediate.
\begin{corollary}
Given mother lattice $\Lambda_V$ and orthonormal basis $\Phi$, let $\beta_{min}$  be the minimum volume contraction ratio  of the child lattices $\Lambda_S$, taken  over all $k$-subsets of $\{1,2,\ldots,n\}$, and $\Delta^\#$ the geometric mean of the packing densities. Then
\begin{equation}
(\beta_{min} \Delta(\Lambda_V))^2 \leq \frac{(\Delta^\#)^2}{{n \choose k}} \leq \frac{\Delta_k(opt)^2}{{n \choose k}}.
\end{equation}
Equality holds in the left inequality iff all the contraction ratios are equal and all the child lattices have equal determinants. Equality holds in the right inequality iff  all child lattices achieve the optimal packing density in dimension $k$.

\end{corollary}

\subsection{Trace Upper Bound}
\begin{theorem}{(Trace Bound)}
For an $(n,k)$ code, the compaction ratio is bounded as
\begin{equation}
\beta_{min}^{2/k} \leq \overline{\beta_S^{2/k}}\leq \frac{k}{n}.
\label{eqn:thmtrace}
\end{equation}
Equality holds if  the shortest vector of each lattice $\Lambda_S$ is the image of the shortest vector in $\Lambda_V$.
\end{theorem}
\begin{proof}
Upon summing over all $k$-subsets $S$ we obtain
\begin{eqnarray}
\sum_S G(\Lambda_S) & = & \sum_S V^t \Phi_S^t \Phi_S V \nonumber \\
& = & V^t \left( \sum_S \Phi_S^t \Phi_S \right) V \nonumber \\
& = & V^t {n-1 \choose k-1} \Phi^t \Phi V \nonumber \\
& = &  {n-1 \choose k-1} V^t V.
\label{eqn:GphiSum}
\end{eqnarray}
By definition the smallest packing radius of any child lattice ${\rho_{min}}$ satisfies 
\begin{equation}
{\rho_{min}}^2  \leq  \rho(\Lambda_S)^2 \leq (1/2) u^tG(\Lambda_S)u
\end{equation}
for any non-zero $u \in \mathbb{Z}^k$ and any $k$-subset $S$.
Upon averaging over subsets $S$ we obtain the upper bound
\begin{eqnarray}
\overline{\rho(\Lambda_S)^2}  & \leq    & \frac{1}{2{n \choose k}}\sum_S u^t G(\Lambda_S) u  \nonumber \\
& = & \frac{1}{2{n \choose k}} u^t \sum_SG(\Lambda_S) u  \nonumber \\
& = & \frac{{n-1 \choose k-1}}{2{n \choose k}}  u^t G(\Lambda_V) u.
\end{eqnarray}
Equality holds if  $\Phi_S V u$ is the shortest vector in $\Lambda_S$ for all $S$.
Thus
\begin{equation}
\overline{\rho(\Lambda_S)^2} \leq  \frac{k}{n} \rho^2(\Lambda_V)
\end{equation}
and (\ref{eqn:thmtrace}) follows immediately.
\end{proof}

\section{Analysis of  Some $(4,k)$ Codes}
\label{sec:fourdimcodes}
We construct $\Phi$ for $n=4$ for various values of $k$ and various mother lattices $\Lambda_V$. Numerical results for the $(4,k)$, $k=2,3$ are presented in Fig.~\ref{fig:plot4dim}, in which $\beta_{min}^{2/k}$ is plotted as a function of the packing density of the mother lattice. We have plotted the determinant bound using both the optimal and the cubic lattice for the child lattices. We have also plotted the trace bound. 

Observe that in the $(n,k)=(4,2)$ case there is a significant gap between the best possible construction and the upper bounds. In the $(4,3)$ case performance close to the  determinant bound is achieved by setting the mother lattice to be the cubic lattice. Also with the cubic lattice as the mother lattice, since the trace bound is lower than the determinant bound, this is proof that it is impossible to simultaneously achieve the packing density of $D_3$ when the mother lattice is the cubic lattice. 

  \begin{figure}[h] 
   \centering
   \includegraphics[width=8.0cm]{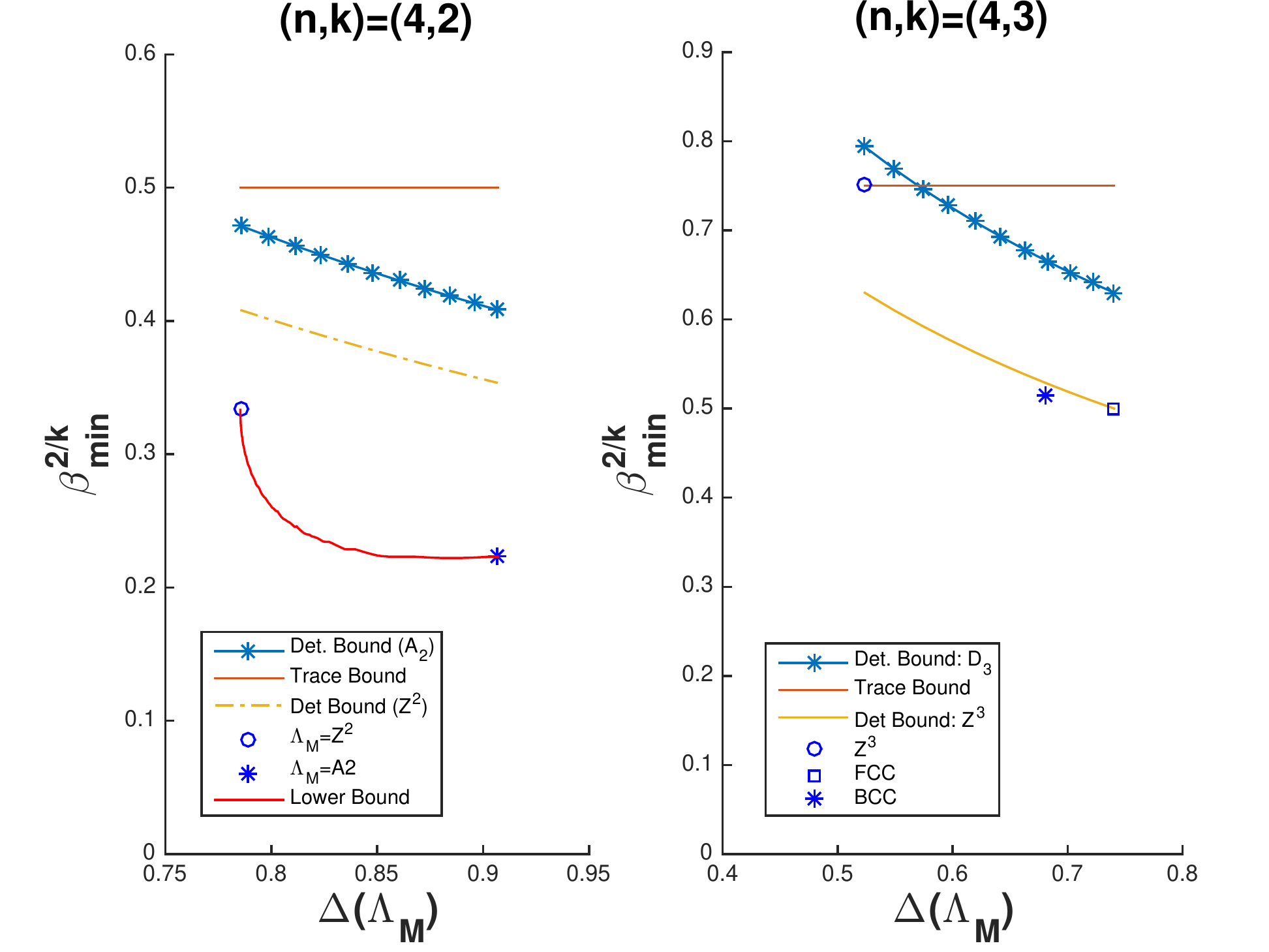} 
   \caption{Bounds on $\beta^{2/k}$ and values obtained from the construction.}
   \label{fig:plot4dim}
\end{figure}

\begin{figure}[htbp] 
   \centering
   \begin{tabular}{cc}
  \includegraphics[width=4.5cm]{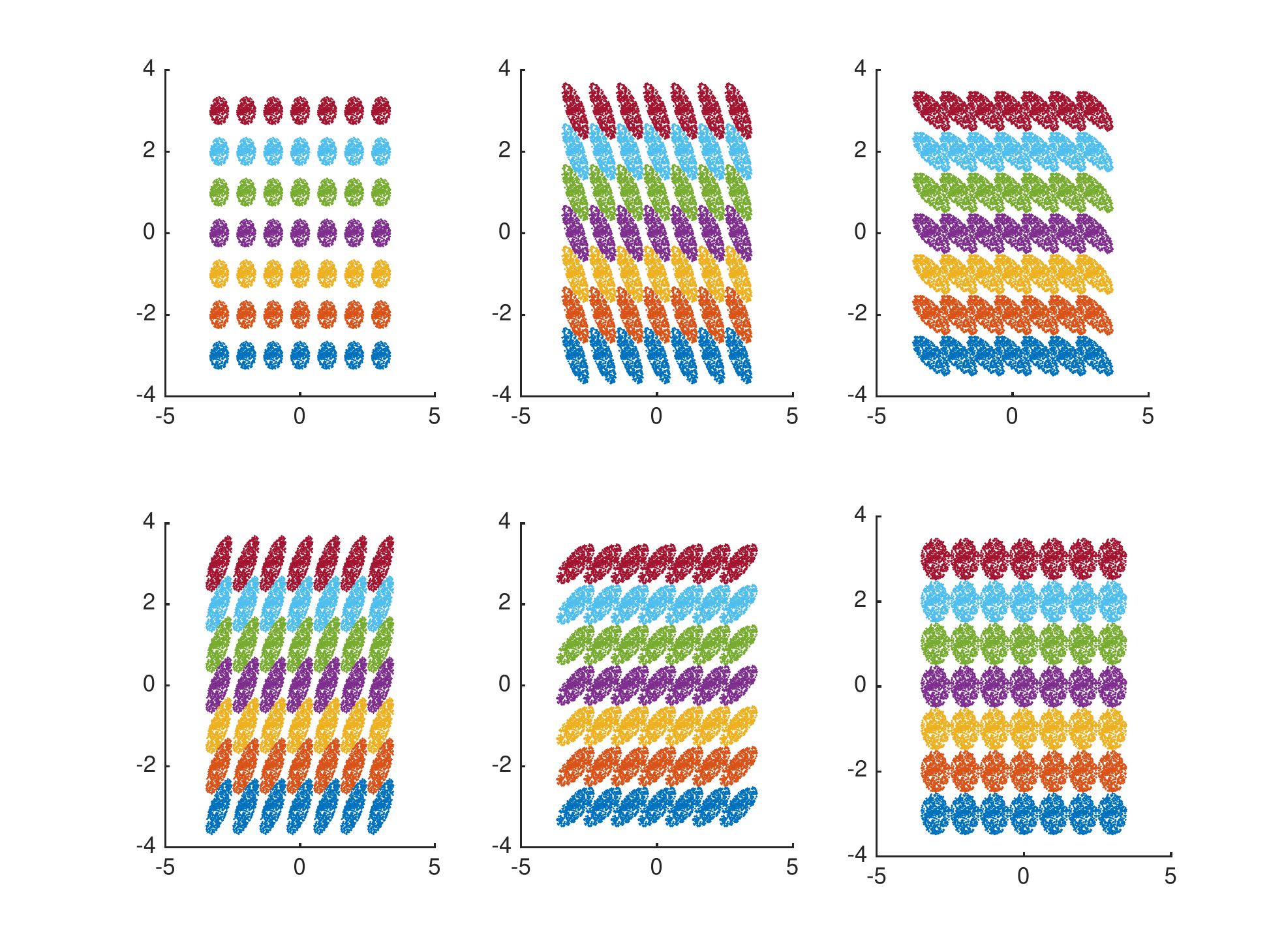}  & (a) \\
\includegraphics[width=4.5cm]{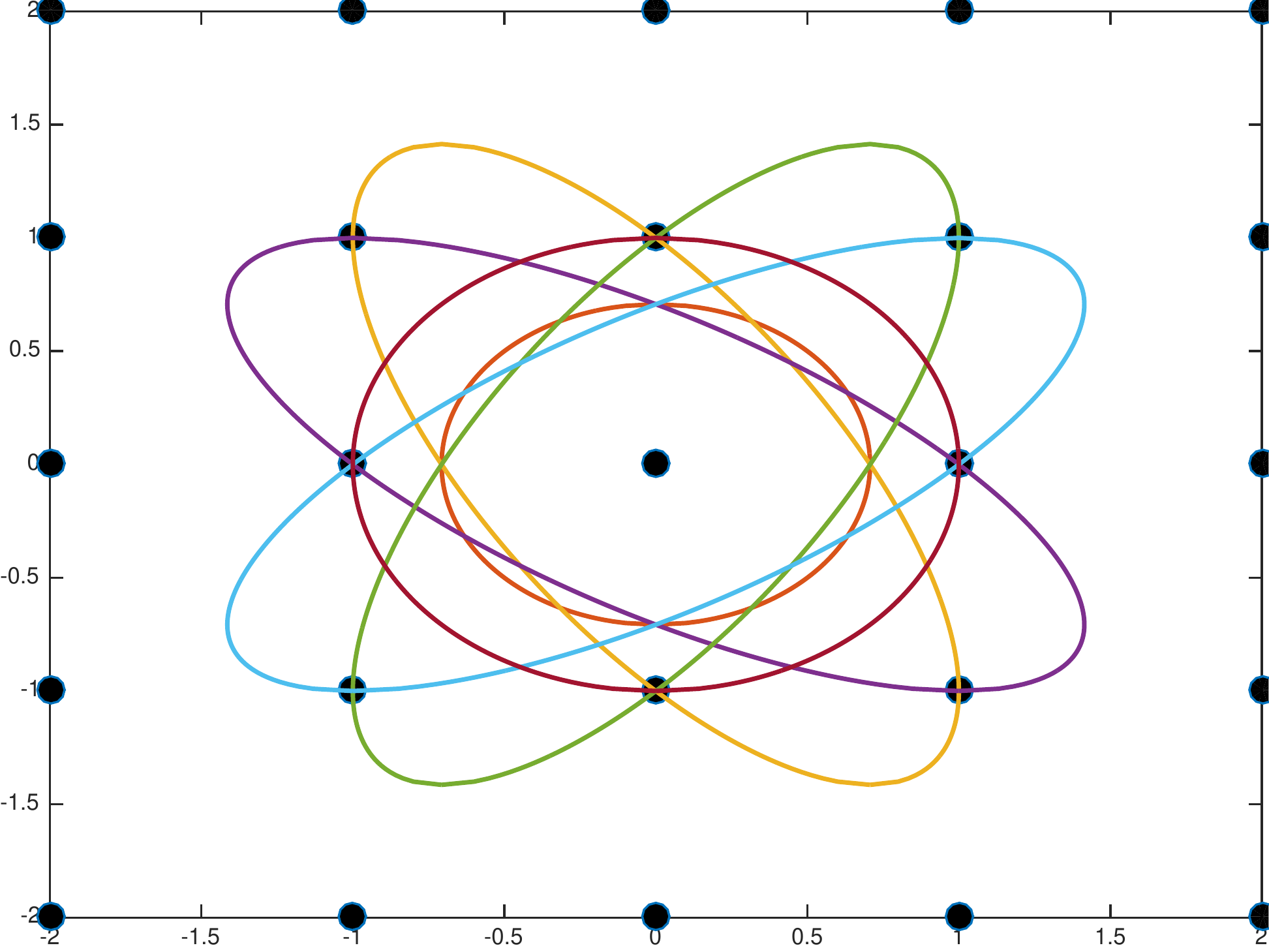} & (b) 
 \end{tabular}

  \caption{Packings (a) and the star body (b) derived from noise spheres in each of the six subspaces for the $(4,2)$ code with $\Lambda_V=\mathbb{Z}^2$ as described in the text.}
   \label{fig:admpack}
\end{figure}
Geometrically, the ability to decode correctly, post-erasure,  with iid Gaussian noise is determined by the largest noise sphere which can be packed by the projected (or child)  lattice $\Lambda_S$ in each of the ${n \choose k}$ subspaces ${\mathcal C}_S$. When each noise sphere is projected back onto the subspace spanned by the columns of $\Phi$, a noise sphere is transformed into an ellipsoid. To see this consider the noise sphere $\|x\|^2\leq r^2$ in subspace ${\mathcal C}_S$. Setting $x=\Phi_S y$, this leads to the noise ellipsoid $B_S(r)=\{y~:~\|\Phi_S y\|^2 \leq r^2\}$. The packing of the noise ellipsoids $B_S(\rho_{min})$, (recall that $\rho_{min}$ is the half the length of the shortest non-zero vector in $\Lambda_S$), by the mother lattice $\Lambda_V$ is shown in the six panels in Fig.~\ref{fig:admpack}(a), one for each subspace. Fig.~\ref{fig:admpack}(b) shows the star body  $B(2\rho_{min})=\bigcup_{s \in \mathcal S}B_S(2\rho_{min})$. This illustrates the star body, which is the union of the six ellipses, two of which are circles. Also shown are  points of the lattice $\Lambda_V$, which in this case is an admissible lattice for this star body and illustrates also the interpretation as a code design problem for the compound channel. Observe that $\mathbb{Z}^2$ is \emph{simultaneously} good as a packing for all six erasure configurations, i.e. for each of the bodies $B_S(\rho_{min})$. Also, $\Lambda_V$  is \emph{simultaneously} critical for five of the six bodies $B_S(2\rho_{min})$ (notice that one circle does not touch any of the lattice points).

\subsection{$(4,1)$}
Let $\Phi^t=\begin{pmatrix} a & a & a & a \end{pmatrix}$, $a=1/2$. We obtain $\beta_{min}^2=1/4$. The trace and determinant upper bounds yield $\beta^2 \leq 1/4$. Hence this construction is optimal.

\subsection{$(4,2)$}
With $\Lambda_V=\mathbb{Z}^2$, $a=1/\sqrt{3}$ and
$$\Phi^t=\left( \begin{array}{cccc} a & a & a & 0\\a &-a & 0 & a
 \end{array} \right)$$
 we obtain six child lattices with Gram matrices
 \begin{eqnarray}
 \begin{pmatrix} 2a^2 & 0 \\0 & 2a^2\end{pmatrix}, & 
 \begin{pmatrix} a^2 & a^2 \\a^2 & 2a^2\end{pmatrix}, & 
 \begin{pmatrix} 2a^2 & a^2 \\a^2 & a^2\end{pmatrix}, \nonumber \\
 \begin{pmatrix} a^2 & -a^2 \\-a^2 & 2a^2\end{pmatrix}, &
 \begin{pmatrix} 2a^2 & -a^2 \\-a^2 & a^2\end{pmatrix}, & 
 \begin{pmatrix} a^2 & 0 \\0 & a^2\end{pmatrix}.
 \end{eqnarray}
All six child lattices are similar to $\mathbb{Z}^2$. The first one has shortest vector of square length $2a^2$ and all the others have square length $a^2$.  This code achieves $\beta_{min}^{2/k}=\beta_{min}=1/3$, $\beta^\#=2^{1/6}/3=0.374$. 
 
The trace upper bound is  $\beta_{min} \leq \overline{\beta_S} \leq 1/2$, regardless of the mother lattice, while the determinant upper bound depends on the mother lattice, and is $\sqrt{2}/3$ and $1/\sqrt{6}=0.408$ for $\Lambda_V=A_2$ (hexagonal lattice) and $\mathbb{Z}^2$, resp. Thus the determinant bound is tighter than the trace bound but greater than $1/3$. This construction does not meet the trace bound or the determinant bound with equality. However with $\Lambda_V=\mathbb{Z}^2$, the following theorem shows this to be the best  $\beta_{min}$ possible. A computer-based search has also failed to reveal any improvements in $\beta^\#$.

\begin{theorem}
Let $\Phi$ be a $4 \times 2$ matrix with orthonormal columns, so that $\Phi^t \Phi=I$. Let $\Lambda_V=\mathbb{Z}^2$. 
Let $\Lambda_S=\Phi_S\Lambda_V$ where $S$ is a  2-subset  of $\{1,2,3,4\}$. Let  $r_S$ be the length of the shortest non-zero vector in $\Lambda_S$. Then $r_S^2 \leq 1/3$ for at least one of the six  2-subsets of $\{1,2,3,4\}$.
\end{theorem}
\begin{proof}
For  this proof we will write
\begin{equation}
\Phi^t= \begin{pmatrix} 
x_1 & x_2 & x_3 & x_4 \\y_1 & y_2 & y_3 & y_4
\end{pmatrix}.
\end{equation}
Our proof is by contradiction. Suppose that the shortest non-zero vector  in all $\Lambda_S$ has square length $r^2>1/3$. Now $x_i^2 < r^2/2$, in at most one position $i$, else the shortest would be smaller than $r^2$. The same is true for $y_i^2$, $(x_i-y_i)^2$ and $(x_i+y_i)^2$. Hence there exists one position $i$, w.l.o.g. $i=1$, such that $x_i^2 \geq r^2/2$ , $y_i^2  \geq r^2/2$  and $ (x_i-y_i)^2 \geq r^2/2$. It follows that (i) $x_1$ and $y_1$ must be of opposite sign and (ii) $(x_1+y_1)^2 <  r^2/2$. (i) is true because  if $x_1$, $y_1$ are of the same sign and  $ (x_1-y_1)^2 \geq r^2/2$ then $|y_1| \geq \sqrt{2} r$ or $ |x_1| \geq  \sqrt{2} r$. Assuming $x_1 \geq \sqrt{2} r$, we have $r^2 \leq  x_3^2+x_4^2 =1-x_1^2-x_2^2  \leq  1 - 2r^2$ which contradicts the hypothesis that $r^2 > 1/3$. (ii) is true for if not $|x_1| \geq  \sqrt{2} r$ or $|y_1| \geq \sqrt{2} r$, and by the same argument as in (i) $r^2$ cannot exceed $1/3$.

Since $(x_1+y_1)^2 < r^2/2$, it follows that $(x_i+y_i)^2 \geq  r^2/2$ for positions $i=2,3,4$. In at least one of these positions, say $i=2$, $x_i^2  \geq  r^2/2$ and $y_i^2  \geq  r^2/2$. Now $x_2$ and $y_2$ must be of the same sign and $(x_2-y_2)^2 <  r^2/2$, by a proof similar to that used before. Thus for $i=3,4$, $(x_i+ y_i)^2 \geq  r^2/2$ and $(x_i- y_i)^2 \geq  r^2/2$. Again for $i=3,4$, both $x_i^2 \geq  r^2/2$ and $y_i^2 \geq  r^2/2$ cannot hold for the same $i$, hence, either $x_3^2 <  r^2/2,~y_3^2 \geq  r^2/2$ and $x_4^2 \geq r^2/2,~y_4^2 <  r^2/2$ or $x_4^2 <  r^2/2,~y_4^2 \geq  r^2/2$ and $x_3^2 \geq r^2/2,~y_3^2 <  r^2/2$. We assume the first case. The proof for the other case is similar.

We have already proved that $x_2$ and $y_2$ are of the same sign.  Assume they are both positive (if not reverse signs of all elements of $\Phi$). Further, assume that $y_2 > x_2$ and consider positions $i=2,3$ (if $x_2 > y_2$, then the same proof applies but for positions $i=2,4$).  We now break up the proof into two cases: 

\flushleft
{\bf Case 1}:  ({\bf $y_3$ and $x_3$ of the same sign}): Either (a) $(y_3-x_3)^2 \geq  r^2$ or (b) $(y_3-x_3)^2 < r^2$.  If (a) 
then $|y_3|-|x_3| \geq  r$ and $y_2^2+y_3^2  \geq x_2^2 + (|x_3|+r)^2 \geq $ $x_2^2+x_3^2+r^2 \geq 2r^2$. Thus $r^2 \leq y_1^2+y_4^2=$ $1-(y_2^2+y_3^3)$ $ \leq 1-2r^2$, hence $r^2 \leq 1/3$. 

If (b) then let $(y_3-x_3)^2 =r^2-\epsilon^2$ ($0< \epsilon^2 \leq r^2/2$). Thus $|y_3|-|x_3|=\sqrt{r^2-\epsilon^2}$ and it follows that $y_3^2 \geq x_3^2 + r^2-\epsilon^2$. Since $(y_2-x_2)^2+(y_3-x_3)^2 \geq r^2$ it follows that $(y_2-x_2)^2 \geq \epsilon^2$ and that $|y_2 | \geq  |x_2|+\epsilon$ and thus $y_2^2  \geq  x_2^2+\epsilon^2$. Thus $y_2^2+y_3^2 \geq x_2^2 +\epsilon^2 +x_3^2 +r^2 -\epsilon^2$ $ =x_2^2+x_3^2 + r^2 \geq 2r^2$. But then $r^2 \leq y_1^2+y_4^2=1-(y_2^2+y_3^2) \leq 1-2r^2$ and it follows that  $r^2 \leq 1/3$.
\flushleft
{\bf Case 2}: ({\bf$y_3$ and $x_3$ are of opposite signs}): Either (a) $(y_3+x_3)^2 \geq r^2$ or (b) $(y_3+x_3)^2  < r^2$. If (a) then $|y_3+x_3| \geq  r$ and due to opposite signs $|y_3|-|x_3| \geq r$ from which $y_3^2 \geq x_3^2+r^2$. It follows that $y_2^2+y_3^2 \geq x_2^2+x_3^2 +r^2$ $ \geq 2r^2$. This implies $r^2 \leq 1/3$. If (b) then let $(y_3+x_3)^2 =r^2-\epsilon^2$. Since $y_3$ and $x_3$ have opposite signs $(|y_3|-|x_3|)^2 =r^2-\epsilon^2$ and thus $y_3^2 \geq x_3^2 +r^2-\epsilon^2$. Since $(y_2+x_2)^2 + (y_3+x_3)^2 \geq r^2$ it follows that $(y_2+x_2)^2+r^2-\epsilon^2 \geq r^2$ and hence $(y_2+x_2)^2 \geq \epsilon^2$ which implies that $y_2^2 \geq x_2^2 +\epsilon^2$. Thus $y_2^2+y_3^2 \geq x_2^2+\epsilon^2 + x_3^2 +r^2-\epsilon^2 \geq 2r^2$. Once again this means $r^2 \leq 1/3$.

\end{proof}

\subsection{$(4,3)$}
It is checked by direct evaluation that for $a=1/{2}$, $\Lambda_V=\mathbb{Z}^3$ and
\begin{equation}
\Phi^t=\begin{pmatrix}  a  & -a & -a & -a \\a & a & -a & a \\a  & -a & a & a
\end{pmatrix},
\end{equation}
the four child lattices have Gram matrices
\begin{eqnarray}
\begin{array}{cc}
\begin{pmatrix} 
3/4       &      1/4       &      1/4     \\
       1/4   &          3/4     &       -1/4     \\
       1/4      &      -1/4      &       3/4   
\end{pmatrix}, & 
\begin{pmatrix}
 3/4      &      -1/4      &       1/4     \\
      -1/4      &       3/4    &         1/4     \\
       1/4      &       1/4      &       3/4    
\end{pmatrix} \nonumber \\
\begin{pmatrix}
3/4        &     1/4       &     -1/4     \\
       1/4    &         3/4     &        1/4     \\
      -1/4    &         1/4       &      3/4  
 \end{pmatrix} & 
\begin{pmatrix}
 3/4      &      -1/4      &     -1/4    \\ 
      -1/4     &       3/4     &      -1/4     \\
      -1/4       &    -1/4       &     3/4   
\end{pmatrix}.
\end{array}
\end{eqnarray}
All child lattices are congruent to the body-centered cubic lattice with packing density $\pi \sqrt{3}/8$. This construction achieves the trace bound $r^2=3/4$ and is optimal. 

Set $\Lambda_V$ to be the face-centered cubic lattice, the densest lattice packing in $\mathbb{R}^3$, packing density  $\Delta_V=\pi/\sqrt{18}=0.7408$, so that
$\det \Lambda_V = 4$,
\begin{equation}
\begin{array}{cc}
V=\begin{pmatrix} -1 & 1 & 0 \\-1 & -1 & 1 \\ 0 & 0 &  -1\end{pmatrix},  & G(\Lambda_V)=
\begin{pmatrix} 
2 & 0 & 1\\0 & 2 & 1\\1 & 1 & 2
\end{pmatrix}
\end{array}.
\end{equation}
With $\Phi$ such that 
 \begin{equation}
 \Phi*V=\begin{pmatrix} 1 &  0  & 0 \\ 2/3 & 1 & 1\\ -2/3  & 1 & 0\\ 1/3 & 0 & 1 \end{pmatrix},
 \end{equation}
we get $\beta^{2/k}=(\rho_{min}/\rho(\Lambda_V))^2=0.5$. 
Each child lattice has unit determinant and shortest vector of length equal to unity.
All four child lattices have an identical packing density of $\pi/6$, which is the packing density of the cubic lattice $\mathbb{Z}^3$.
Also observe that the determinant bound (\ref{eqn:DetBoundTh}) holds with equality, which implies that the selected $\Phi$ is optimal among rotations for which all child lattices achieve the packing density of the cubic lattice.

Let $\Lambda_V$ the body-centered cubic lattice, $\Delta(\Lambda_V)=\pi \sqrt{3}/8$, $\rho^2=3/4$,
$$V=   \begin{pmatrix} 1 &     -1 &    1 \\
    -1  &    1  &   1 \\
    -1  &  -1   &  1 \end{pmatrix}, G(\Lambda_V)= \begin{pmatrix} 3 & -1 & -1\\-1 & 3 & -1\\-1 & -1 & 3 \end{pmatrix}.
    $$
Upon setting $a=1/2$, $b=a+\sqrt{a}$, $c=a-\sqrt{a}$ and 
    $$\Phi=\begin{pmatrix}
     \sqrt{a} &   a    &     0 \\
         0  &   a   &  -\sqrt{a} \\
         0  &   a  &   \sqrt{a} \\
    \sqrt{a}  &   -a   &       0
    \end{pmatrix},
    \Phi*V=\begin{pmatrix}-c & c & b\\-c & b & c\\-b & c & b\\b & -b & -c \end{pmatrix}
    $$
 we obtain child lattices with Gram matrices given in terms of $d=2c^2+b^2$, $e=2b^2+c^2$, $f=-c^2-2bc$, $g=c^2+2bc$, $h=3bc$, $i=-b^2-2bc$, 
    \begin{eqnarray}
    \begin{pmatrix} d & f & -g \\ f & d & h\\-g & h & e \end{pmatrix}, & 
     \begin{pmatrix} d & -g & f \\ -g & e & -h\\f & -h & d \end{pmatrix}, \nonumber \\
       \begin{pmatrix} e & -g & i \\ -g & d & -h\\i & -h & e \end{pmatrix}, & 
        \begin{pmatrix} e & -i & -g \\ -i & e & -h\\-g & -h & d \end{pmatrix}.
    \end{eqnarray}
  Each child lattice has determinant $4$ and shortest vector with square length $d= (9/4-1/\sqrt{2})=1.5429$ and packing density  $\Delta=\pi d^{3/2} /12=0.5017$. This    results in  $\beta^{2/k}=d/3=0.5143$. 
 
 \section{Summary}
 \label{sec:summary}
We considered  the design of power-constrained rank-$k$ lattice codes in $\mathbb{R}^n$, with the property that error free recovery is possible for any $k$ code symbols and the minimum distance for each of the ${n \choose k}$ lattices obtained by projection onto $k$-dim subspaces spanned by $k$ coordinate vectors are bounded from below.  The potential application is to erasure channels with additive noise. Bounds on the performance are derived and the performance of specific constructions are investigated for $n=4$.

\bibliography{../../../erasures,../../../broadcast} 
\bibliographystyle{abbrv}

\end{document}